\documentclass[a4paper,12pt]{article}
\title{Hermitian hull-variation of vector rank-metric codes and self-orthogonal generalized Gabidulin codes}
\author{Duy Ho \\
Department of Mathematical Sciences, \\
UAE University, PO Box 15551, Al Ain, UAE\\
Email: duyho92@gmail.com}
\date{}
\usepackage{parskip}
\usepackage{todonotes}
\usepackage{amsthm,amsmath,amssymb}
\usepackage{url}
\usepackage{booktabs}
\usepackage{enumitem}
\allowdisplaybreaks
\setenumerate[0]{label=(\roman*)}

\begin{document}
\maketitle
\theoremstyle{plain}
\newtheorem{lemma}{Lemma}[section]
\newtheorem{theorem}[lemma]{Theorem}
\newtheorem{corollary}[lemma]{Corollary}
\newtheorem{proposition}[lemma]{Proposition}
\theoremstyle{definition}
\newtheorem{definition}[lemma]{Definition}
\newtheorem{remark}[lemma]{Remark}
\newtheorem{example}[lemma]{Example}
\newtheorem{question}[lemma]{Open problem}
\newcommand{\eps}{\varepsilon}
\newcommand{\inprod}[1]{\left\langle #1 \right\rangle}
\newcommand{\la}{\lambda}
\newcommand{\al}{\alpha}
\newcommand{\om}{\omega}
\newcommand{\gam}{\gamma}
\newcommand{\be}{\beta}
\newcommand{\sig}{\sigma}
\newcommand\rank{\mathrm{rank}}
\newcommand\nullity{\mathrm{nullity}}
\newcommand{\F}{\mathbb{F}}
\newcommand{\E}{\mathcal{E}}
\newcommand{\M}{\mathcal{M}}
\newcommand{\Tr}{\mathrm{Tr}}
\newcommand{\GL}{\mathrm{GL}}
\newcommand{\Z}{\mathbb{Z}}

\begin{abstract}
We study the Hermitian hull-variation problem for vector rank-metric codes. 
Except for one parameter pair, we show that the Hermitian hull dimension of such a code can be reduced to any smaller value within its equivalence class, and in particular every such code is equivalent to a Hermitian LCD code. 
We then address the existence of maximum rank distance (MRD) codes with prescribed Hermitian hull dimension. To this end, we introduce the notion of a \emph{scaled trace-self-dual basis} of a finite field extension, which exists in all cases, and use it to construct Hermitian self-orthogonal generalized Gabidulin codes for every prime power. Combined with the hull-variation theorem, this yields MRD codes attaining every admissible Hermitian hull dimension.
\end{abstract}



\section{Introduction}

Let $C$ be a linear code over a finite field, equipped with either the Euclidean or the Hermitian inner product on its ambient space. The \emph{hull} of $C$ is the intersection
\[
H(C) := C \cap C^\perp,
\]
where $C^\perp$ denotes the dual of $C$ with respect to the chosen inner product. The hull is a fundamental subcode of $C$, and its dimension governs several constructions. 
Codes with trivial hull, called \emph{linear complementary dual} (LCD) codes, were used by Carlet and Guilley~\cite{carlet2016} as countermeasures against side-channel attacks.
In 2018, Guenda, Jitman, and Gulliver~\cite{guenda2018} showed that the hull dimension of $C$ equals the number of maximally entangled pairs required to construct an entanglement-assisted quantum error-correcting code (EAQECC) from $C$, in the framework introduced by Brun, Devetak, and Hsieh~\cite{brun2006}.
Understanding the behavior of the hull dimension under code equivalence is thus a natural and well-motivated problem.

In~\cite{carlet2018}, Carlet, Mesnager, Tang, Qi, and Pellikaan proved that every $[n,k]$ linear code over $\F_q$ is equivalent to a Euclidean LCD code provided $q > 3$, and that every $[n,k]$ linear code over $\F_{q^2}$ is equivalent to a Hermitian LCD code provided $q > 2$. The latter result was sharpened in~\cite{chen2023, ling2024}. There it was shown that for every linear code $C$ over $\F_{q^2}$ ($q > 2$) with Hermitian hull dimension $h$ and every $\ell \in \{0, 1, \ldots, h\}$, there exists a code equivalent to $C$ with Hermitian hull dimension exactly $\ell$. 
This phenomenon leads to the so-called hull-variation problem formulated by Hao Chen in 2023, see \cite{chen2023-2}. The problem is stated as follows. 

\bigskip
\textit{Hull-Variation Problem:} When a linear code $C$ is transformed to an equivalent linear code $C'$, how is its Euclidean or Hermitian hull changed?
\bigskip

In parallel with these developments in the Hamming metric, the rank metric, introduced by Delsarte~\cite{delsarte1978} and Gabidulin~\cite{gabidulin1985}, has emerged as a robust framework with applications in network coding~\cite{koetter2008, silva2008} and, more recently, in quantum error correction. Delfosse and Z\'emor~\cite{delfosse2024} propose quantum Gabidulin codes via a CSS-type construction to correct circuit faults in a stacked quantum memory, requiring a self-dual normal basis of $\F_{2^n}/\F_2$ and thereby restricting the memory to $n \times n$ with $n$ odd. 
Nizuka and Matsumoto~\cite{matsumoto2026} subsequently give an alternative construction over $\F_{2^{2m}}/\F_2$ that uses Hermitian self-orthogonal Gabidulin codes in the framework of Matsumoto and Uyematsu~\cite{matsumoto2000}.
This removes the odd-length restriction and improves the relative rank distance.
The Hermitian inner product on vector rank-metric codes is thus a natural object of study in this setting, and motivates an analogue of the hull-variation problem for vector rank-metric codes.

In~\cite{ho2026}, the author and Johnsen investigated the Euclidean hull-variation problem for vector rank-metric codes. The present paper first addresses the analogous Hermitian hull-variation problem. We show that, outside the single case $(q,n) = (2,2)$, every Hermitian hull dimension up to $\dim H(C)$ is attained within the equivalence class of $C$. In particular, every such code is equivalent to a Hermitian LCD code.
 
Among rank-metric codes, those attaining the Singleton-like bound on minimum rank distance are called \emph{maximum rank distance (MRD)} codes, and play the role of MDS codes in the rank metric. The second part of the paper addresses the existence of MRD codes with prescribed Hermitian hull dimension. Islam and Horlemann~\cite{islam2022} studied the Galois hull dimensions of Gabidulin codes and, in particular, constructed Hermitian self-orthogonal Gabidulin codes whenever a trace-self-dual $\F_q$-basis of $\F_{q^{2m}}$ exists. By Seroussi--Lempel~\cite{seroussi1980} and Jungnickel--Menezes--Vanstone~\cite{jungnickel1990}, such a basis exists if and only if $q$ is even, leaving the case $q$ odd open. We resolve this by introducing the notion of a \emph{scaled trace-self-dual basis}. We show that such a basis exists for every pair $(q, m)$, and we use it to construct Hermitian self-orthogonal generalized Gabidulin codes for every prime power $q$. Combined with the hull-variation theorem above, this gives MRD codes in $\F_{q^{2m}}^{2m}$ with parameters $[2m, k, 2m-k+1]_{q^{2m}/q}$ and Hermitian hull dimension $\ell$, for every $(q,m) \neq (2,1)$, every $1 \le k \le m$, and every $0 \le \ell \le k$.

The paper is organized as follows. In Section~\ref{sec:prelim} we recall the necessary background. Section~\ref{sec:hullvariation} establishes the Hermitian hull-variation theorem and identifies the obstruction in the case $(q, n) = (2, 2)$. Section~\ref{sec:prescribed-hull} constructs Hermitian self-orthogonal MRD codes from scaled trace-self-dual bases and combines them with the hull-variation construction to obtain MRD codes with each prescribed Hermitian hull dimension. Section~\ref{sec:examples} provides explicit examples over $\F_{16}$ and $\F_9$. In Section~\ref{sec:euclidean-comparison}, we discuss the analogous questions for the Euclidean inner product, complete a case left open in~\cite{ho2026}, and observe that, in contrast to the Hermitian setting, Euclidean self-orthogonal MRD codes do not exist in even characteristic.

\section{Preliminaries}\label{sec:prelim}

 Let $q$ be a prime power  and let $m\geq 1$. Let
$ K := \F_{q^{2m}},   F := \F_{q^m}. $ 
The \emph{trace} and \emph{norm} of $K/\F_q$ are the maps
\[
\Tr_{K/\F_q} : K \to \F_q, \qquad \Tr_{K/\F_q}(x) := \sum_{i=0}^{2m-1} x^{q^i},
\]
\[
N_{K/\F_q} : K \to \F_q, \qquad N_{K/\F_q}(x) := \prod_{i=0}^{2m-1} x^{q^i} = x^{(q^{2m}-1)/(q-1)}.
\]
Throughout, we will use $\Tr := \Tr_{K/\F_q}$ and $N := N_{K/\F_q}$ unless otherwise noted.

We write $\mathbf{1}_n := (1, 1, \ldots, 1) \in K^n$ for the all-ones vector, and $K\mathbf{1}_n := \{\lambda \mathbf{1}_n : \lambda \in K\}$ for the line it spans. Also, $e_i \in K^n$ denotes the $i$-th standard basis vector. 
\subsection{Vector rank-metric codes and generalized Gabidulin codes}
 
 An $[n,k]_{q^{2m}/q}$ \emph{vector rank-metric code} is a $k$-dimensional $K$-subspace $C\subseteq K^n$ equipped with the rank distance
\[
d_R(u,v) := \dim_{\F_q}\langle u_1-v_1,\ldots,u_n-v_n\rangle_{\F_q}.
\]
The \emph{minimum rank distance} of $C$ is
\[
d_R(C) := \min\{d_R(c, c') : c, c' \in C,\ c \neq c'\}.
\]
A code $C$ with parameters as above and minimum rank distance $d_R(C) = d$ is denoted $[n, k, d]_{q^{2m}/q}$. When $n \leq 2m$, the Singleton-like bound for the rank metric gives $d_R(C) \leq n - k + 1$. A code meeting this bound is \emph{maximum rank distance} (MRD).
Two such codes $C,C'$ are \emph{equivalent}  if there exists $A\in\GL_n(\F_q)$ and generator matrices $G,G'$ of $C,C'$ respectively such that $G' = GA$. This equivalence preserves both the dimension and the rank distance, see for example \cite{gorla2021}.

Let $\alpha = (\alpha_1,\ldots,\alpha_{2m})$ be an $\F_q$-basis of $K$, and let $s$ be a positive integer with $\gcd(s,2m)=1$. The \emph{generalized Gabidulin code} $G_{k,s}(\alpha)$ is the $K$-linear code of length $2m$ and dimension $k$ generated by
\[
G^{(s)} = \begin{pmatrix} \alpha \\ \alpha^{q^s} \\ \vdots \\ \alpha^{q^{s(k-1)}} \end{pmatrix},
\qquad
\alpha^{q^a} := (\alpha_1^{q^a}, \ldots, \alpha_{2m}^{q^a}).
\]
 The construction was introduced by Kshevetskiy and Gabidulin~\cite{kshevetskiy2005}, which generalizes  Gabidulin's original construction~\cite{gabidulin1985} from $s=1$ to general $s$ coprime to $2m$.

\begin{lemma}[\cite{kshevetskiy2005}, see also~\cite{gabidulin1985}]\label{lem:GabMRD}
For any $\F_q$-basis $\alpha$ of $K$ and any $s$ with $\gcd(s,2m)=1$, the code $G_{k,s}(\alpha)$ is MRD with parameters $[2m,k,2m-k+1]_{q^{2m}/q}$.
\end{lemma}

The case $s=1$ recovers the classical Gabidulin code $G_k(\alpha) := G_{k,1}(\alpha)$. 

\begin{remark}\label{rem:gengabi-equiv}
Fix $k$ and $s$ with $\gcd(s, 2m) = 1$. For any two $\F_q$-bases $\alpha, \alpha'$ of $K$ with $\alpha' = \alpha A$, $A \in \GL_{2m}(\F_q)$,   we have $G_{k,s}(\alpha') = G_{k,s}(\alpha) \cdot A$. Hence all generalized Gabidulin codes with fixed parameters $(k, s)$ form a single  rank-metric equivalence class, and the MRD property is independent of $\alpha$. 
\end{remark}

\subsection{Hermitian inner product and Hermitian hull}\label{ssec:hermitian}

The Hermitian involution $\sigma : K \to K$ is defined by $\sigma(x) := x^{q^m}$.   
For $x,y\in K^n$, the \emph{Hermitian inner product} is defined as
\[
\langle x,y\rangle_H := \sum_{i=1}^n x_i\,\sigma(y_i).
\]
The \emph{Hermitian dual} of $C$ is $C^{\perp_H} := \{x\in K^n: \langle x,c\rangle_H = 0\text{ for all }c\in C\}$, and the \emph{Hermitian hull} is $H(C) := C\cap C^{\perp_H}$. A code with $H(C)=\{0\}$ is called \emph{Hermitian LCD}.

For a matrix $G\in K^{k\times n}$, we write $G^\dagger$ for the $n\times k$ matrix with entries $(G^\dagger)_{ij} = \sigma(G_{ji})$. For matrices $A,B$ over $K$, $(AB)^\dagger = B^\dagger A^\dagger$. 

\begin{lemma} \label{lem:rankhullH}  \label{rankhullH}
Let $C$ be an $[n,k]_{q^{2m}/q}$ vector rank-metric code with generator matrix $G$. 
Let $\ker(GG^\dagger) := \{z \in K^{k \times 1} : (GG^\dagger)z = 0\}$.
The map $z \mapsto z^\dagger G$ defines a $\sigma$-semilinear bijection between $\ker(GG^\dagger)$ and $H(C)$. In particular,
\[
\dim_{K} H(C) = \nullity(GG^\dagger) = k - \rank_{K}(GG^\dagger).
\]
\end{lemma}

\begin{proof}
A vector $v \in K^n$ is in $C$ if and only if $v = z^\dagger G$ for some $z \in K^{k\times 1}$, and such a $z$ is unique since $G$ has rank $k$. Such $v$ is in $H(C) = C \cap C^{\perp_H}$ if and only if additionally $vG^\dagger = 0$, equivalently $(GG^\dagger) z = 0$. Hence $z \mapsto z^\dagger G$ defines a $\sigma$-semilinear bijection $\ker(GG^\dagger) \to H(C)$. Since $\sigma$ is an automorphism of $K$, this bijection preserves $K$-dimension, and the dimension formula follows by applying the rank--nullity theorem to $GG^\dagger$.
\end{proof}

\subsection{Scaled trace-self-dual basis}

The trace bilinear form $B_1(x,y) := \Tr_{K/\F_q}(xy)$ is a nondegenerate symmetric $\F_q$-bilinear form on $K$, viewed as an $\F_q$-vector space of dimension $2m$. 
An $\F_q$-basis $\alpha = (\alpha_1, \ldots, \alpha_{2m})$ of $K$ is \emph{(trace-)self-dual} (with respect to $B_1$) if $\Tr_{K/\F_q}(\alpha_i\alpha_j) = \delta_{ij}$ for all $i,j$. 
By \cite{seroussi1980,jungnickel1990}, a self-dual basis of $\F_{q^r}/\F_q$ exists if and only if $q$ is even or both $q$ and $r$ are odd. 
In our setting $r = 2m$ is even, so no self-dual basis exists when $q$ is odd, which motivates the scaled variant that we  introduce next.

For any $\lambda\in K^*$, the \emph{scaled trace bilinear form} is
\[
B_\lambda(x,y) := \Tr_{K/\F_q}(\lambda xy).
\]
Equivalently, $B_\lambda(x,y) = B_1(\mu_\lambda(x), y)$, where $\mu_\lambda : K \to K$, $x \mapsto \lambda x$.
Since $\mu_\lambda$ is an $\F_q$-linear automorphism of $K$ and $B_1$ is nondegenerate, $B_\lambda$ is also a nondegenerate symmetric bilinear form.

\begin{definition}
A \emph{scaled trace-self-dual basis} of $K/\F_q$ is a pair $(\alpha, \lambda)$, where $\alpha = (\alpha_1, \ldots, \alpha_{2m})$ is an $\F_q$-basis of $K$ and $\lambda\in K^*$, such that
\[
\Tr_{K/\F_q}(\lambda\alpha_i\alpha_j) = \delta_{ij} \quad\text{for all } 1\leq i,j\leq 2m.
\]
When $\lambda = 1$, the basis $\alpha$ is an ordinary  self-dual basis.
\end{definition}

\begin{lemma}\label{lem:scaled-basis}
For every prime power $q$ and every positive integer $m$, there exists a scaled trace-self-dual basis of $K = \F_{q^{2m}}$ over $\F_q$.
\end{lemma}

\begin{proof} 
The extension degree of $K/\F_q$ is $2m$, which is always even. 
If $q$ is even, then by~\cite{wan2003} (see also~\cite{jungnickel1990}), an ordinary  self-dual basis of $K/\F_q$ exists, corresponding to $\lambda=1$. In the remainder of the proof, we consider the case $q$ is odd.
In this case, no ordinary trace-self-dual basis of $K/\F_q$ exists, so we cannot take $\lambda = 1$. 

Fix an arbitrary $\F_q$-basis $\beta = (\beta_1, \ldots, \beta_{2m})$ of $K$, and let $M_1 := (\Tr(\beta_i\beta_j))_{ij}\in\F_q^{2m\times 2m}$ be the Gram matrix of the form $B_1$ in this basis. Since $B_1$ is nondegenerate, $\Delta := \det(M_1)\in\F_q^*$.

For $\lambda\in K^*$, let $L_\lambda$ be the matrix representation of the map $\mu_\la$ in basis $\beta$, see \cite[Example 13.2.10]{mullen2013}.  
The identity $B_\lambda(x,y) = B_1(\mu_\lambda(x), y)$ gives 
\[
M_\lambda := (\Tr(\lambda\beta_i\beta_j))_{ij} = L_\lambda^\top \, M_1.
\]
Since $\det(L_\lambda) = N_{K/\F_q}(\lambda)$ (see for example \cite[Exercises 2.25 and 2.26]{lidl1997}), we have
\[
\det(M_\lambda) = N_{K/\F_q}(\lambda)\Delta.
\] 
Fix $\lambda$ with $N_{K/\F_q}(\lambda)=\Delta^{-1}$, so that $\det(M_\lambda)=1$. 
By \cite[Lemma 1]{jungnickel1990} (see also \cite[p. 143]{artin1957}),  the form $B_\lambda$ is equivalent to the standard form $\sum_i x_iy_i$. 
In particular, $B_\lambda$ admits an orthonormal basis and so there exists a basis $\alpha = (\alpha_1,\ldots,\alpha_{2m})$ of $K$ over $\F_q$ such that $B_\lambda(\alpha_i, \alpha_j) = \Tr(\lambda\alpha_i\alpha_j) = \delta_{ij}$.
\end{proof}

\begin{remark}
The scaled trace-self-dual basis is closely related to the notion of \emph{weakly self-dual basis} studied by Geiselmann and Gollmann~\cite{geiselmann1993}. A basis $\alpha$ is weakly self-dual if its trace-dual basis $\{\beta_i\}$ (with $\Tr(\alpha_i\beta_j) = \delta_{ij}$) satisfies $\beta_i = \lambda\alpha_{\pi(i)}$ for some $\lambda\in K^*$ and permutation $\pi$. Lemma~\ref{lem:scaled-basis} corresponds to the special case $\pi = \mathrm{id}$. The existence problem of weakly self-dual bases does not seem to be addressed in the  literature. For the case  $\pi = \mathrm{id}$, we have provided a self-contained proof above.
\end{remark}

\section{Hermitian hull-variation of equivalent vector rank-metric codes}\label{sec:hullvariation}
The main results of this section are  Theorems~\ref{thm:rankone-h} and  \ref{thm:binary-constant-hull}, which together state that, except when $(q, n) = (2, 2)$,
for every vector rank-metric code $C$ with $\dim_K H(C) = h\geq 1$
there exists an equivalent code $C'$ with $\dim_K H(C') = h - 1$. 
The proof relies on   Lemmas~\ref{lem:rankone} and \ref{lem:descent-matrix} for the case  $n\geq 3$.

The case $n = 2$ is treated separately. For $q\geq 4$, the
construction in the proof of Lemma~\ref{lem:descent-matrix} does not use the condition $n \ge 3$ and  still applies. 
The case $(q, n) = (3, 2)$ requires a different argument, 
which we provide in Lemma~\ref{lem:descent-n2-q3}. 
For the remaining case $(q, n) = (2, 2)$, we   show in Remark~\ref{rem:obstruction-22} that hull-variation is not always possible.

Theorems~\ref{thm:rankone-h} and  \ref{thm:binary-constant-hull} allow us to obtain the full hull-variation classification, which we state in Theorem~\ref{thm:hermitianmain}: for every pair
$(q, n)\neq(2, 2)$, every intermediate hull dimension
$\ell\in\{0, 1, \ldots, h\}$ is attained by some equivalent code. 
In particular, we obtain Corollary~\ref{cor:hermitianLCD} showing that every vector rank-metric code with $(q,n) \ne (2,2)$ is equivalent to a Hermitian LCD code.
 
\begin{lemma} \label{lem:rankone}
Let $B\in K^{k\times k}$ be a Hermitian matrix of nullity $h\geq 1$. Let $a\in K^{k\times 1}$ and $\lambda\in\F_q^*$. If there exists $z\in\ker(B)$ with $z^\dagger a\neq 0$, then $B + \lambda aa^\dagger$ has nullity $h-1$.
\end{lemma}

\begin{proof}
Since $B$ is Hermitian and $z\in\ker B$, we have $z^\dagger B = (Bz)^\dagger = 0$. Assume that $(B+\lambda aa^\dagger)x = 0$ for some $x\in K^{k\times 1}$. Multiplying on the left by $z^\dagger$, we obtain
\[
0 = z^\dagger(B + \lambda aa^\dagger)x = \lambda(z^\dagger a)(a^\dagger x).
\]
Since $\lambda\neq 0$ and $z^\dagger a\neq 0$, it follows that $a^\dagger x = 0$, and so $Bx = 0$. Conversely, any $x\in\ker(B)$ with $a^\dagger x = 0$ lies in $\ker(B + \lambda aa^\dagger)$. Hence,
\[
\ker(B + \lambda aa^\dagger) = \ker(B)\cap\{x: a^\dagger x = 0\}.
\]
Note that the hyperplane $H:=\{a^\dagger x = 0\}$ does not contain $z$. From the dimension formula, we have
\[
\dim(\ker(B) \cap H) \;=\; \dim\ker(B) + \dim H - \dim(\ker(B) + H) \;=\; h + (k-1) - k \;=\; h-1.
\]
Therefore, $B + \lambda aa^\dagger$ has nullity $h - 1$.
\end{proof}



\begin{lemma}\label{lem:descent-matrix}
Let $q$ be a prime power and $n \ge 3$.  Let $w\in K^n$ be nonzero. If $q = 2$,  further assume that $w\notin K\mathbf 1_n$. 
Then there exist $u\in\F_q^n$ and $M\in\GL_n(\F_q)$ such that 
\begin{enumerate}
\item $wu^\top \neq 0$,
\item $MM^\top = I_n + \lambda u^\top u$ for some $\lambda\in\F_q^*$.
\end{enumerate}
\end{lemma}
\begin{proof}
We split into cases depending on $q$.

\paragraph{Case $q\geq 4$.} Choose $i\in\operatorname{supp}(w)$. Let  $u := e_i$, so that $wu^\top= w_i \ne 0$. Let $\mu\in\F_q^*$ with $\mu^2\neq 1$. Such $\mu$ exists since $q \ge 4$. Let $\lambda := \mu^2 - 1\in\F_q^*$. Let $M\in\GL_n(\F_q)$ be the diagonal matrix $\mathrm{diag}(1,\ldots,1,\mu,1,\ldots,1)$ with $\mu$ in the $i$-th position. Then $MM^\top = I_n + \lambda e_i^\top e_i$.

\bigskip

For the cases $q=2$ and $q=3$, there is no $\mu\in\F_q^*$ with $\mu^2\neq 1$, so we will describe the choices of $u$ and $M$ differently from the case $q \ge 4$. 

\paragraph{Case $q = 3$.} Choose $i\in\mathrm{supp}(w)$, and pick any $j,l\in\{1,\ldots,n\}\setminus\{i\}$ with $j\neq l$, which is possible since $n\geq 3$. Consider the two vectors
\[
u^{(1)} = e_i + e_j + e_l, \qquad u^{(2)} = -e_i + e_j + e_l.
\]
We have $u^{(t)}(u^{(t)})^\top = 0$   for $t = 1, 2$. Also,
\[
w(u^{(1)})^\top - w(u^{(2)})^\top = 2w_i \neq 0.
\]
 Hence $w(u^{(t)})^\top \neq 0$ for some $t\in\{1, 2\}$. We fix such a $t$ and let $u := u^{(t)}$. Then (i) holds.

Let  $M := I_n + 2u^\top u$. We note that $M$ is symmetric. We have $\det(M) = 1 + 2uu^\top = 1\neq 0$, so $M\in\GL_n(\F_3)$. Furthermore, since $uu^\top = 0$, we have $(u^\top u)^2 = u^\top(uu^\top)u = 0$. Then 
\[
MM^\top = M^2 = I_n + 4u^\top u + 4(u^\top u)^2 = I_n + u^\top u,
\]
and so (ii) holds. 

\paragraph{Case $q = 2$.} Since $w\notin K\mathbf 1_n$, there exist coordinates $i,j$ satisfying $w_i\neq w_j$. Let $u := e_i + e_j\in\F_2^n$. Then $wu^\top = w_i + w_j\neq 0$ and so (i) holds.

Let
\[
L :=
\begin{bmatrix}
1 & 1 & 1 \\
0 & 1 & 1 \\
1 & 0 & 1
\end{bmatrix}\in\GL_3(\F_2).
\]
A direct computation gives
\[
LL^\top
=
\begin{bmatrix}
1&0&0\\
0&0&1\\
0&1&0
\end{bmatrix}
=
I_3+\widetilde u^\top\widetilde u,
\qquad
\widetilde u:=(0,1,1)\in\F_2^3.
\]
Now let
\[
\widehat u:=(0,1,1,0,\ldots,0)\in\F_2^n, \qquad
D:=\operatorname{diag}(L,I_{n-3})\in\GL_n(\F_2).
\]
Since \(n\geq 3\), we can choose \(l\in\{1,\ldots,n\}\setminus\{i,j\}\).  Let $\pi\in S_n$ be any permutation with $\pi(1) = l$, $\pi(2) = i$, $\pi(3) = j$, and let $P$ be the corresponding permutation matrix, so $Pe_r^\top = e_{\pi(r)}^\top$ for all $r$. Then $P\widehat u^\top = e_i^\top + e_j^\top = u^\top$.
Define
\[
M:=PDP^\top\in\GL_n(\F_2).
\]
 Then
\[
MM^\top = PDD^\top P^\top = P(I_n + \widehat u^\top\widehat u)P^\top = I_n + u^\top u,
\]
using $PP^\top = I_n$ and $P\widehat u^\top = u^\top$, so (ii) holds.
 \end{proof}


\begin{lemma}\label{lem:descent-n2-q3}
Let $K = \F_{9^m}$. Let $C\subseteq K^2$ be a vector rank-metric code with
  $\dim_K H(C) = 1$. Then there exists $M\in\GL_2(\F_3)$ such that
  $CM$ is Hermitian LCD.
\end{lemma}

\begin{proof}
Since \(C\subseteq K^2\) has nonzero Hermitian hull, we must have
\(\dim_K C=1\) and \(H(C)=C\). Let $C=Kv$
for some nonzero vector
$v=(x,y)\in K^2.$
Since \(H(C)=C\), we have   $\langle v,v\rangle_H=0$, and so
\[
x\sigma(x)+y\sigma(y)=0.
\]
In particular, \(x,y\neq 0\).
For \(t\in\F_3\), let
\[
M_t=
\begin{bmatrix}
1&0\\
t&1
\end{bmatrix}
\in\GL_2(\F_3).
\]
Then $ vM_t=(x+ty,y) $, and
\[
\begin{aligned}
\langle vM_t,vM_t\rangle_H
&=(x+ty)\sigma(x+ty)+y\sigma(y)\\
&=x\sigma(x)+y\sigma(y)
+t\bigl(x\sigma(y)+y\sigma(x)\bigr)
+t^2y\sigma(y)\\
&=
t\bigl(x\sigma(y)+y\sigma(x)\bigr)
+t^2y\sigma(y).
\end{aligned}
\]
Furthermore, since $y\sigma(y)\neq 0$, the polynomial
$P(t):=\langle vM_t,vM_t\rangle_H$
is a quadratic. Since \(\F_3\) has three elements, there exists \(t\in\F_3\) such that
$P(t)\neq 0,$ and so with the choice $M:=M_t$, we have $\langle vM,vM\rangle_H \ne 0$. Then the code $C'=CM$ is 1-dimensional with 
$H(C') \ne C'$, so it is Hermitian LCD.
\end{proof}

\begin{theorem} \label{thm:rankone-h}
Let $q$ be a prime power and let $C$ be an
$[n,k]_{q^{2m}/q}$ vector rank-metric code with $\dim_K H(C) = h\geq 1$.
Assume that $(q,n)\neq(2,2)$, and that $H(C)\neq K\mathbf 1_n$ when $q=2$. 
Then there exists an equivalent code $C'$ with $\dim_K H(C') = h - 1$.
\end{theorem}
\begin{proof} Let $G$ be a generator matrix of $C$.
Let $w\in H(C)$ be nonzero. If $q = 2$,  we   choose $w$ such  that $w\notin K\mathbf 1_n$. 
By Lemma~\ref{lem:rankhullH}, there exists
$z\in\ker(GG^\dagger)$ with $z^\dagger G = w$.

\paragraph{Case $n\geq 3$.} By Lemma~\ref{lem:descent-matrix}, there exist
$u\in\F_q^n$ and $M\in\GL_n(\F_q)$ with $wu^\top\neq 0$ and
$MM^\top = I_n + \lambda u^\top u$ for some $\lambda\in\F_q^*$. Set
$a := Gu^\top$ and $G' := GM$, so
\[
G'(G')^\dagger
= GMM^\top G^\dagger
= GG^\dagger + \lambda a a^\dagger.
\]
Let $C'$ be the code with generator matrix $G'$.
Since $z\in\ker(GG^\dagger)$ and $z^\dagger a = wu^\top\neq 0$,
by Lemma~\ref{lem:rankone}, we have $\dim_K H(C') = h - 1$.

\paragraph{Case $n = 2$, $q\geq 4$.} 
We follow the case $q \ge 4$ in the proof of
Lemma~\ref{lem:descent-matrix}. 
Choose $i\in\operatorname{supp}(w)$, take $u := e_i$ and let $M$ be the
$2\times 2$ diagonal matrix with $\mu\in\F_q^*$ ($\mu^2\neq 1$) in
position $i$ and $1$ in the other position. 
Then $wu^\top = w_i\neq 0$ and $MM^\top = I_2 + (\mu^2 - 1)u^\top u$. 
Similar to the previous case, the code $C':=CM$ is equivalent to $C$ with $\dim_K H(C') = h - 1$.

\paragraph{Case $n = 2$, $q = 3$.} 
Since $h \geq 1$ and $k \leq n - 1 = 1$, we have $h = k = 1$. By
Lemma~\ref{lem:descent-n2-q3}, there exists $M \in \GL_2(\F_3)$ such that
$C':=CM$ is Hermitian LCD.
\end{proof}

\begin{theorem}
\label{thm:binary-constant-hull}
Let \(K=\F_{2^{2m}}\)  and let \(C\) be an
\([n,k]_{2^{2m}/2}\) vector rank-metric code with \(n\geq 3\). 
Assume that $H(C)=K\mathbf{1}_n$.
Then  there exists an equivalent code $C'$ such that $\dim_K H(C')=0.$
\end{theorem}
\begin{proof}
Since \(\mathbf{1}_n\in H(C)\), it is self-orthogonal. Then
$ 0=\langle \mathbf{1}_n,\mathbf{1}_n\rangle_H=n $
in \(K\). So \(n\) is even. 

1. Let $G$ be a generator matrix of $C$. Under suitable row operations, we can assume that
\[
G=
\begin{bmatrix}
1 & \mathbf{1}_{n-1}\\
0 & B
\end{bmatrix}.
\]
Since \(\mathbf{1}_n\in C^{\perp_H}\), we have $\mathbf{1}_{n-1}B^\dagger=0.$
Moreover,
\[
GG^\dagger
=
\begin{bmatrix}
0 & 0\\
0 & BB^\dagger
\end{bmatrix}.
\]
Since \(\dim_K H(C)=1\), the nullity of \(GG^\dagger\) is \(1\). Hence \(BB^\dagger\) is
nonsingular.

2.  Define
\[
M=
\begin{bmatrix}
1 & \mathbf{1}_{n-1}\\
0 & I_{n-1}
\end{bmatrix}
\in \GL_n(\F_2).
\]
Then
\[
GM
=
\begin{bmatrix}
1 & \mathbf{1}_{n-1}\\
0 & B
\end{bmatrix}
\begin{bmatrix}
1 & \mathbf{1}_{n-1}\\
0 & I_{n-1}
\end{bmatrix}
=
\begin{bmatrix}
1 & 0\\
0 & B
\end{bmatrix}.
\]
Let $C'$ be the code with generator matrix
$G':=GM$. Then $C'$ is equivalent to $C$. Furthermore,
\[
G'(G')^\dagger
=
\begin{pmatrix}
1 & 0\\
0 & BB^\dagger
\end{pmatrix},
\]
which  is nonsingular because \(BB^\dagger\) is nonsingular. Therefore
$
\dim_K H(C')=0.
$
\end{proof}

From Theorems~\ref{thm:rankone-h} and  \ref{thm:binary-constant-hull}, we obtain the following. 
\begin{theorem}\label{thm:hermitianmain}
Let $q$ be a prime power,  and let $C$ be an
$[n,k]_{q^{2m}/q}$ vector rank-metric code with
$\dim_K H(C) = h$. Assume that $(q, n) \neq (2, 2)$. Then for every
$\ell \in \{0, 1, \ldots, h\}$ there exists an equivalent code $C'$
with $\dim_K H(C') = \ell$.
\end{theorem}
\begin{corollary}\label{cor:hermitianLCD}
Let $q$ be a prime power. Assume that $(q, n) \neq (2, 2)$.
Then every $[n, k]_{q^{2m}/q}$ vector rank-metric code is equivalent to a
Hermitian LCD code.
\end{corollary}

\begin{remark}\label{rem:obstruction-22} We discuss the case \((q,n)=(2,2)\). Let
\(K=\F_4=\F_2(\omega)\), where \(\omega^2+\omega+1=0\). Let $v=(1,\om)$ and consider the 1-dimensional
code $ C=Kv\subseteq K^2. $
We have
$ \langle v,v\rangle_H  = 0,  $
and so \(C\subseteq C^{\perp_H}\). In particular,  $ H(C)=C $
and \(\dim_K H(C)=1\).

We claim that no equivalent code is Hermitian LCD. Let
$M = \begin{bmatrix} a & b \\ c & d \end{bmatrix}\in\GL_2(\F_2)$, where
$a, b, c, d\in\F_2$ and $ad - bc = 1$. Then
\[
vM= (1,\omega)M
= (a + c\omega,\, b + d\omega).
\]
Since $M \in\GL_2(\F_2)$, we have $a + c\omega \ne 0$ and $b + d\omega \ne 0$.   Hence,
\[
\langle vM,\,vM\rangle_H
= (a + c\omega)^3 + (b + d\omega)^3
= 1 + 1 = 0.
\]
 Consequently, the
  code $CM = \langle vM\rangle_K$ has
$\dim_K H(CM) = 1$ for every $M\in\GL_2(\F_2)$, and $C$ has no
Hermitian LCD code in its equivalence class.

This example shows that hull-variation is generally not possible for the case $(q,n)=(2,2)$.
\end{remark}

\section{MRD codes with prescribed Hermitian hull dimension}\label{sec:prescribed-hull} 
 
We recall that a scaled trace-self-dual basis  of $K/\F_q$ is a pair $(\alpha, \lambda)$, where $\alpha = (\alpha_1, \ldots, \alpha_{2m})$ is an $\F_q$-basis of $K$ and $\lambda\in K^*$, such that
\[
\Tr_{K/\F_q}(\lambda\alpha_i\alpha_j) = \delta_{ij} \quad\text{for all } 1\leq i,j\leq 2m.
\]
When $\lambda = 1$, the basis $\alpha$ is an ordinary trace-self-dual basis. 
\begin{lemma}\label{lem:Sr}
Let $(\alpha, \lambda)$ be a scaled trace-self-dual basis of $K = \F_{q^{2m}}$ over $\F_q$. For $0\leq r\leq 2m-1$, define
\[
S_r := \sum_{i=1}^{2m} \alpha_i\,\alpha_i^{q^r} \in K.
\]
Then $S_0 = \lambda^{-1}$ and $S_r = 0$ for $1\leq r\leq 2m-1$.
\end{lemma}

\begin{proof}
Since $\Tr(\lambda\alpha_i\alpha_j) = \delta_{ij}$, the family $(\lambda\alpha_1, \ldots, \lambda\alpha_{2m})$ is the trace-dual basis of $(\alpha_1, \ldots, \alpha_{2m})$ in $K/\F_q$. Therefore every $x\in K$ admits the expansion
\[
x = \sum_{i=1}^{2m} \alpha_i\,\Tr(\lambda\alpha_i x).
\]
Expanding the trace as $\Tr(y) = \sum_{r=0}^{2m-1} y^{q^r}$ and rearranging,
\begin{align*}
x &= \sum_{i=1}^{2m} \alpha_i \sum_{r=0}^{2m-1} (\lambda\alpha_i x)^{q^r} \\
  &= \sum_{r=0}^{2m-1} \lambda^{q^r}\,\Bigl(\sum_{i=1}^{2m} \alpha_i\,\alpha_i^{q^r}\Bigr)\,x^{q^r} \\
  &= \sum_{r=0}^{2m-1} \lambda^{q^r}\,S_r\,x^{q^r}.
\end{align*}
Here, the polynomials $P(X)=X$ and $Q(X)=\sum_{r=0}^{2m-1} \lambda^{q^r}\,S_r\,X^{q^r}$ are both elements in $K[X]$ and have degree at most $q^{2m-1}$. Since $P(x)=Q(x)$ for every $x \in K$,  they coincide in $K[X]$. Matching the coefficients, we have 
\[
\lambda S_0 = 1, \qquad \lambda^{q^r} S_r = 0 \quad\text{for } 1 \leq r \leq 2m-1.
\]
Hence $S_0 = \lambda^{-1}$ and, since $\lambda\neq 0$, $S_r = 0$ for $r\geq 1$.
\end{proof}

\begin{theorem}\label{thm:so-mrd}
Let $(\alpha,\lambda)$ be a scaled trace-self-dual basis of $K/\F_q$, and let $s$ be a positive integer with $\gcd(s, 2m) = 1$. For every $1\leq k\leq m$, the generalized Gabidulin code
$G_{k,s}(\alpha)$
is a Hermitian self-orthogonal MRD code with parameters $[2m, k, 2m-k+1]_{q^{2m}/q}$.
\end{theorem}

\begin{proof}
Let $C := G_{k,s}(\alpha)$.
By Lemma~\ref{lem:GabMRD}, $C$ has parameters $[2m, k, 2m-k+1]_{q^{2m}/q}$ and is MRD.
It remains to show that $C\subseteq C^{\perp_H}$.

For $0 \le i \le k-1$, the row $v_i$ of the generator matrix $G^{(s)}$ is
\[
v_i  = \alpha^{q^{is}} = (\alpha_1^{q^{is}}, \ldots, \alpha_{2m}^{q^{is}}).
\]
For $0 \le i, j \le k-1$, the Hermitian inner product of $v_i$ and $v_j$ is
\[
\langle v_i, v_j\rangle_H
= \sum_{t=1}^{2m}\alpha_t^{q^{is}}\,\sigma(\alpha_t^{q^{js}})
= \sum_{t=1}^{2m}\alpha_t^{q^{is}}\,\alpha_t^{q^{js+m}}
= \Bigl(\sum_{t=1}^{2m} \alpha_t\,\alpha_t^{q^{(j-i)s+m}}\Bigr)^{q^{is}}
= \bigl(S_{(j-i)s+m}\bigr)^{q^{is}}.
\]
By Lemma~\ref{lem:Sr}, $S_r = 0$ unless $r \equiv 0 \pmod{2m}$. So
$\langle v_i, v_j\rangle_H = 0$ unless $(j-i)s + m \equiv 0 \pmod{2m}$, equivalently
\[
(j-i)s \equiv m \pmod{2m}.
\]
Since $\gcd(s, 2m) = 1$ and $2m$ is even, $s$ is odd, so $sm \equiv m \pmod{2m}$. Then
\[
(j-i)s \equiv m \pmod{2m} \iff j - i \equiv m \pmod{2m},
\]
which is not possible, since $\lvert j - i \rvert \le k - 1 \le m - 1$.
Therefore the rows of $G^{(s)}$ are mutually Hermitian-orthogonal, and so $C\subseteq C^{\perp_H}$.
\end{proof}

\begin{corollary}\label{cor:so-mrd-self-dual}
For every prime power $q$ and every $m\geq 1$, there exists a Hermitian self-dual MRD code in $\F_{q^{2m}}^{2m}$ with parameters $[2m, m, m+1]_{q^{2m}/q}$.
\end{corollary}

From Theorems \ref{thm:hermitianmain} and \ref{thm:so-mrd}, we have the following. 
\begin{theorem}\label{thm:combined}
Assume that $(q, m) \neq (2, 1)$. Let $1 \leq k \leq m$ and $0 \leq \ell \leq k$.
Then there exists an MRD code $C \subseteq K^{2m}$ with parameters
$[2m, k, 2m-k+1]_{q^{2m}/q}$ and Hermitian hull dimension $\ell$.
\end{theorem} 
 
\begin{remark}\label{rem:weak-so-sidorenko}
The   identities $S_r = 0$ for $1 \le r \le 2m-1$ in Lemma~\ref{lem:Sr} are  the scaled Hermitian analogue of a condition appearing in the decoding literature. Jerkovits, Sidorenko, and Wachter-Zeh~\cite{jerkovits2021} call an $\F_q$-basis $\alpha$ of $\F_{q^n}$ a \emph{weak self-orthogonal basis} if the Moore matrix $M_n(\alpha)$ satisfies $M_n(\alpha)\,M_n(\alpha)^\top = D$ for some diagonal $D$, and they use such bases as code locators for Gabidulin codes in a decoder for space-symmetric rank errors. After re-indexing, their condition is equivalent to the vanishing of $\sum_l \alpha_l\,\alpha_l^{q^r}$ for all $1 \le r \le n-1$, which is the Euclidean specialisation $\lambda = 1$ of Lemma~\ref{lem:Sr}. 
\end{remark}


\section{Examples}\label{sec:examples} 
In this section, we provide three examples to demonstrate the machinery used in Sections~\ref{sec:hullvariation} and~\ref{sec:prescribed-hull}.
Examples~\ref{ex:binary-constant-hull} and~\ref{ex:rankone-descent-h1} illustrate the hull-variation construction of Section~\ref{sec:hullvariation} for  the binary case. 
Example~\ref{ex:binary-constant-hull} treats a code with $H(C) = K\mathbf 1_n$, covered by Theorem~\ref{thm:binary-constant-hull}, and Example~\ref{ex:rankone-descent-h1} treats a code with $H(C) \neq K\mathbf 1_n$, covered by Theorem~\ref{thm:rankone-h}. 

Example~\ref{ex:F9-G1} demonstrates the construction of a scaled trace-self-dual basis of $\F_9/\F_3$ and uses it to produce a Hermitian self-orthogonal generalized Gabidulin code in $\F_9^2$, illustrating Theorem~\ref{thm:so-mrd}.

\begin{example}\label{ex:binary-constant-hull}
Let \(q=2\), \(m=2\), so \(K=\F_{2^4}=\F_{16}\).
Let $\om \in K$ be such that $\omega^4+\omega+1=0$ and identify
$ K=\F_2(\omega).$
Then 
\[
\sigma(\omega)=\omega+1,\qquad \sigma(\omega+1)=\omega.
\]
Let \(n=4\), and consider the \([4,2]_{16/2}\) vector rank-metric code \(C\subseteq K^4\)
with generator matrix
\[
G=
\begin{bmatrix}
1 & 1 & 1 & 1 \\
0 & 1 & \omega & \omega+1
\end{bmatrix}.
\]
Let $ B=(1,\omega,\omega+1)\in K^{1\times 3}. $ 
Using \(\sigma(\omega)=\omega+1\) and \(\sigma(\omega+1)=\omega\), we get
\[
BB^\dagger
=
1+\omega(\omega+1)+(\omega+1)\omega
=
1+(\omega^2+\omega)+(\omega^2+\omega)
=
1.
\]
Hence,
\[
GG^\dagger
=
\begin{bmatrix}
0 & 0 \\
0 & BB^\dagger
\end{bmatrix}
=
\begin{bmatrix}
0 & 0 \\
0 & 1
\end{bmatrix}.
\]
We see that $(GG^\dagger)e_1^\top = 0$, so $e_1^\top \in \ker(GG^\dagger)$. Applying the bijection $z \mapsto z^\dagger G$ of Lemma~\ref{rankhullH}, we obtain $\mathbf 1_4 = e_1 G \in H(C)$. Therefore, $\dim_K H(C) = 1$ and $H(C) = K\mathbf 1_4$.

Now we  use the
construction from Theorem~\ref{thm:binary-constant-hull}. Set
\[
M=
\begin{bmatrix}
1 & 1 & 1 & 1 \\
0 & 1 & 0 & 0 \\
0 & 0 & 1 & 0 \\
0 & 0 & 0 & 1
\end{bmatrix}
\in \GL_4(\F_2).
\]
Then
\[
G':=GM
=
\begin{bmatrix}
1 & 0 & 0 & 0 \\
0 & 1 & \omega & \omega+1
\end{bmatrix}.
\]
Consequently,
\[
G'(G')^\dagger
=
\begin{bmatrix}
1 & 0 \\
0 & BB^\dagger
\end{bmatrix}
=
\begin{bmatrix}
1 & 0 \\
0 & 1
\end{bmatrix}.
\]
This matrix is nonsingular, so
\[
\dim_K H(CM)=2-\rank_K\bigl(G'(G')^\dagger\bigr)=0.
\]
Therefore \(C':=CM\) is Hermitian LCD and is  equivalent to \(C\).
\end{example}

\begin{example}\label{ex:rankone-descent-h1}
Let \(q=2\), \(m=2\), so \(K=\F_{2^4}=\F_{16}\).
Let $\om \in K$ be such that $\omega^4+\omega+1=0$ and identify
$ K=\F_2(\omega)$ as in Example \ref{ex:binary-constant-hull}. 
Let $n=4$ and consider the \([4,2]_{16/2}\) vector rank-metric code \(C\) with generator matrix
\[
G=
\begin{bmatrix}
\omega^3 & 1 & 0 & 0 \\
0 & 0 & 1 & 0
\end{bmatrix}.
\]
Let $ w=(\omega^3,1,0,0)$ be the first row of \(G\).  Then
\[
\langle w,w\rangle_H
=
\omega^3\sigma(\omega^3)+1
=
1+1
=
0.
\]
The second row \(v=(0,0,1,0)\) satisfies
\[
\langle v,w\rangle_H=0,
\qquad
\langle v,v\rangle_H=1.
\]
Therefore
\[
GG^\dagger=
\begin{bmatrix}
0&0\\
0&1
\end{bmatrix},
\]
and $ \dim_K H(C)=2-\rank_K(GG^\dagger)=1.$ Furthermore, we see that $w \in H(C)$, and so
$ H(C)=Kw \ne K\mathbf 1_4.$ 

We now follow the \(q=2\) construction in Lemma~\ref{lem:descent-matrix}.
Let
\[
u=e_1+e_2=(1,1,0,0)\in\F_2^4.
\]
Then $ wu^\top=\omega^3+1\neq 0.$
Choosing \(l=3\) in the proof of Lemma~\ref{lem:descent-matrix} gives
\[
M=
\begin{bmatrix}
1&1&0&0\\
0&1&1&0\\
1&1&1&0\\
0&0&0&1
\end{bmatrix}
\in \GL_4(\F_2).
\]
A direct computation gives
\[
MM^\top
=
\begin{bmatrix}
0&1&0&0\\
1&0&0&0\\
0&0&1&0\\
0&0&0&1
\end{bmatrix}
=
I_4+u^\top u.
\]

Let \(C':=CM\), with generator matrix
\[
G':=GM
=
\begin{bmatrix}
\omega^3 & \omega^3+1 & 1 & 0\\
1&1&1&0
\end{bmatrix}.
\]
Then
\[
G'(G')^\dagger
=
GMM^\top G^\dagger
=
GG^\dagger+(Gu^\top)(Gu^\top)^\dagger.
\]
Here, 
\[
Gu^\top=
\begin{bmatrix}
\omega^3+1\\
0
\end{bmatrix},
\]
and $(\omega^3+1)\sigma(\omega^3+1)=\omega^2+\omega+1\neq 0. $
Thus
\[
G'(G')^\dagger
=
\begin{bmatrix}
\omega^2+\omega+1&0\\
0&1
\end{bmatrix},
\]
which is nonsingular. Hence $ \dim_K H(C')=2-\rank_K\bigl(G'(G')^\dagger\bigr)=0.$
Therefore \(C'=CM\) is Hermitian LCD and is rank-metric equivalent to \(C\).
\end{example}

\begin{example}\label{ex:F9-G1}
Let $q=3$, $m=1$, so $K = \F_9$. We demonstrate the construction of  a scaled trace-self-dual basis of $\F_9/\F_3$ together with a Hermitian self-orthogonal MRD code of dimension $k=1$ in $\F_9^2$.

Let $\om \in K$ be such that $\omega^2+1=0$ and identify
$ K=\F_3(\omega)$. For an element $a+b\omega \in K$, its trace  is $\Tr(a+b\omega) = 2a$ and its norm is $N(a+b\omega) = a^2 + b^2$.

The Gram matrix of $B_1$ in the basis $(1,\omega)$ is
\[
M_1 = \begin{pmatrix} \Tr(1) & \Tr(\omega) \\ \Tr(\omega) & \Tr(\omega^2)\end{pmatrix} = \begin{pmatrix} 2 & 0 \\ 0 & 1\end{pmatrix}, \qquad \Delta = \det(M_1) = 2.
\]
Following the proof of Lemma~\ref{lem:scaled-basis}, we search for $\lambda \in \F_9^*$ with $N(\lambda) = \Delta^{-1} = 2$. Take $\lambda = 1+\omega$, whose norm is $N(\lambda)= 1 + 1 = 2$.

Let $\alpha := (\omega, 1-\omega)$, which is an $\F_3$-basis of $\F_9$.
We claim that $(\alpha, \lambda)$ is a scaled trace-self-dual basis, i.e. $\Tr(\lambda\alpha_i\alpha_j) = \delta_{ij}$.
Using $\omega^2 = -1$, we compute $\alpha_1^2 = -1$, $\alpha_2^2 = \omega$, and $\alpha_1\alpha_2 = 1+\omega$. Hence,
\[
\lambda \alpha_1^2 = -(1+\omega) = 2+2\omega, \quad
\lambda \alpha_2^2 = (1+\omega)\omega = 2+\omega, \quad
\lambda \alpha_1\alpha_2 = (1+\omega)^2 = 2\omega.
\]
We have $\Tr(\lambda\alpha_1^2) = \Tr(\lambda\alpha_2^2) = 1$ and $\Tr(\lambda\alpha_1\alpha_2) = 0$, as required.

We now describe the Gabidulin code from the basis $(\alpha, \lambda)$ as in Theorem \ref{thm:so-mrd}. Let $s := 1$. By Lemma~\ref{lem:GabMRD}, the  code $G_{1,1}(\alpha) = K\cdot\alpha \subseteq \F_9^2$ is a $[2,1,2]_{9/3}$ MRD code. Since $\sigma(\omega) = -\omega$, we have
\[
\langle \alpha, \alpha\rangle_H = \omega \cdot (-\omega) + (1-\omega)(1+\omega) = -\omega^2 + (1 - \omega^2) = 1 + 2 = 0,
\]
so $G_{1,1}(\alpha)$ is Hermitian self-orthogonal. Since the dimension equals half the length, the code is in fact Hermitian self-dual.
\end{example}

\section{Comparison with the Euclidean inner product}\label{sec:euclidean-comparison}
In this section we consider  analogous results of
Sections~\ref{sec:hullvariation} and~\ref{sec:prescribed-hull} for the
Euclidean inner product. Let $q$ be a prime
power   and let $K:=\F_{q^m}$. Let $n \ge 2$. We recall the standard Euclidean inner product on $K^n$ is given as 
\[
\langle x,y\rangle_E:=\sum_{i=1}^n x_i y_i .
\]
For a vector rank-metric code $C\subseteq K^n$, we define its Euclidean dual $C^{\perp_E}$ and its Euclidean hull $H_E(C)$ as
\[
C^{\perp_E}:=\{x\in K^n:\langle x,c\rangle_E=0
\text{ for all }c\in C\}, \qquad H_E(C):=C\cap C^{\perp_E}. 
\]

The hull-variation results of Section~\ref{sec:hullvariation} carry over to
this setting. Replacing $G^\dagger$ by $G^\top$ throughout the proofs of
Theorems~\ref{thm:rankone-h} and~\ref{thm:binary-constant-hull} gives their
Euclidean analogues, since the equivalence matrices $M\in\GL_n(\F_q)$
already satisfy $M^\dagger=M^\top$, so the  case $n \ge 3$
is handled by the same constructions. The case $n=2$ was handled in \cite[Theorems 3.4 and 3.5]{ho2026}.
We obtain the full Euclidean
hull-variation classification in Theorem~\ref{thm:euclideanmain} below, which
in particular covers the case $q\in\{2,3\}$ with $h\geq 2$ left open
in~\cite{ho2026}, where only the intermediate hull dimensions
$\ell\in\{0,\ldots,h-2\}$ were attained.

\begin{theorem} \label{thm:euclideanmain}
Let $C$ be an $[n,k]_{q^m/q}$ vector rank-metric code with
$\dim_K H_E(C)=h$. 
Then there exists an  $[n,k]_{q^m/q}$ code $C'$ equivalent to $C$ with 
$\dim_K H_E(C')=\ell$ for each $\ell\in\{0,1,\ldots,h\}$.
\end{theorem}

In contrast to the Hermitian setting of
Section~\ref{sec:prescribed-hull}, the analogous construction of MRD codes
with every prescribed Euclidean hull dimension is not  always  possible.
Following the proof of \cite[Theorem 2.1]{nebe2016}, we show that Euclidean self-orthogonal MRD codes do not exist in even characteristic.

\begin{proposition}\label{prop:euclidean-no-so-even}
Let $q$ be even, $K:=\F_{q^m}$, and assume that $1\le n\le m$.
Let $C\subseteq K^n$ be an $\F_{q^m}$-linear MRD code of dimension
$k\ge 1$. Then $C$ is not Euclidean self-orthogonal. 
\end{proposition}

\begin{proof}
Suppose  that $C\subseteq C^{\perp_E}$. 
Then for every codeword $c=(c_1,\ldots,c_n)\in C$, we have
\[
0=\langle c,c\rangle_E
=
\sum_{i=1}^n c_i^2
=
\left(\sum_{i=1}^n c_i\right)^2,
\]
which implies $ \sum_{i=1}^n c_i=0$. Hence,
\[
C\subseteq H:=
\left\{x\in K^n:\sum_{i=1}^n x_i=0\right\}.
\]
The Euclidean dual of $H$ is
$
H^{\perp_E}=K\mathbf 1_n.
$
Since $C\subseteq H$, we get $H^{\perp_E}\subseteq C^{\perp_E},$ 
and so $\mathbf 1_n\in C^{\perp_E}.$

On the other hand, by MRD duality
\cite[Corollary~41]{ravagnani2016}, the Euclidean dual $C^{\perp_E}$ is MRD
with parameters $[n,n-k,k+1]_{q^m/q}.$
Thus every nonzero codeword of $C^{\perp_E}$ has rank at least $k+1\ge 2$,
contradicting the condition $\mathbf 1_n\in C^{\perp_E}$. This proves the proposition.
\end{proof}


\begin{thebibliography}{99}

\bibitem{artin1957}
Artin, E. (1957). Geometric Algebra. Interscience Publishers, New York.

\bibitem{brun2006}
Brun, T.~A., Devetak, I., and Hsieh, M.-H. (2006).
Correcting quantum errors with entanglement. Science, 314(5798), 436--439.



\bibitem{carlet2016}
Carlet, C., and Guilley, S. (2016). Complementary dual codes for counter-measures to side-channel attacks. Adv. Math. Commun., 10, 131--150.

\bibitem{carlet2018}
Carlet, C., Mesnager, S., Tang, C., Qi, Y., and Pellikaan, R. (2018). Linear codes over $\F_q$ are equivalent to LCD codes for $q>3$. IEEE Trans. Inform. Theory, 64(4), 3010--3017.

\bibitem{chen2023}
Chen, H. (2023). New MDS entanglement-assisted quantum codes from MDS Hermitian self-orthogonal codes. Des. Codes Cryptogr., 91, 2665--2676.

\bibitem{chen2023-2}
Chen, H. (2023). On the hull-variation problem of equivalent linear codes. IEEE Trans. Inform. Theory, 69, 2911--2922.

\bibitem{delacruz2021}
de la Cruz, J., Evilla, J. R., and \"Ozbudak, F. (2021). Hermitian rank metric codes and duality. IEEE Access, 9, 38479--38487.

\bibitem{delfosse2024}
Delfosse, N., and Z\'emor, G. (2024). Correction of circuit faults in a stacked quantum memory using rank-metric codes. arXiv:2411.09173.

\bibitem{delsarte1978}
Delsarte, P. (1978). Bilinear forms over a finite field, with applications to coding theory. J. Combin. Theory Ser. A, 25(3), 226--241.




\bibitem{gabidulin1985}
Gabidulin, E. M. (1985). Theory of codes with maximum rank distance. Problemy Peredachi Informatsii, 21(1), 3--16.

\bibitem{geiselmann1993}
Geiselmann, W., and Gollmann, D. (1993). Self-dual bases in $\F_{q^n}$. Des. Codes Cryptogr., 3, 333--345.


\bibitem{gorla2021}
Gorla, E. (2021). Rank-metric codes. In Concise Encyclopedia of Coding Theory, Chapman and Hall/CRC, pp. 227--250.


\bibitem{guenda2018}
Guenda, K., Jitman, S., and Gulliver, T.~A. (2018).
Constructions of good entanglement-assisted quantum error correcting codes. Des. Codes Cryptogr., 86(1), 121--136.

\bibitem{ho2026}
Ho, D., and Johnsen, T. (2026). On the hull-variation problem of equivalent vector rank-metric codes. Adv. Math. Commun., 22, 163--174.

\bibitem{islam2022}
Islam, H., and Horlemann, A.-L. (2023). Galois hull dimensions of Gabidulin codes. In Proc. 2023 IEEE Int. Symp. Inform. Theory (ISIT), Taipei, pp. 1834--1839. (arXiv:2211.05068, 2022).

\bibitem{jerkovits2021}
Jerkovits, T., Sidorenko, V., and Wachter-Zeh, A. (2021).
Decoding of space-symmetric rank errors. In Proc. 2021 IEEE Int. Symp. Inform. Theory (ISIT), Melbourne, pp.~658--663.


\bibitem{jungnickel1990}
Jungnickel, D., Menezes, A. J., and Vanstone, S. A. (1990). On the number of self-dual bases of $\mathrm{GF}(q^m)$ over $\mathrm{GF}(q)$. Proc. Amer. Math. Soc., 109(1), 23--29.

\bibitem{koetter2008}
Koetter, R., and Kschischang, F. R. (2008). Coding for errors and erasures in random network coding. IEEE Trans. Inform. Theory, 54(8), 3579--3591.

\bibitem{kshevetskiy2005}
Kshevetskiy, A., and Gabidulin, E. M. (2005). The new construction of rank codes. In Proc. IEEE Int. Symp. Inform. Theory (ISIT), Adelaide, pp. 2105--2108.

\bibitem{lidl1997}
Lidl, R., and Niederreiter, H. (1997). Finite Fields, 2nd ed. Encyclopedia of Mathematics and its Applications, vol. 20, Cambridge University Press, Cambridge.

\bibitem{ling2024}
Luo, G., Ezerman, M. F., Grassl, M., and Ling, S. (2024). Constructing quantum error-correcting codes that require a variable amount of entanglement. Quantum Inf. Process., 23, Paper No. 4.

\bibitem{matsumoto2000}
Matsumoto, R., and Uyematsu, T. (2000).
Constructing quantum error-correcting codes for $p^m$-state systems from classical error-correcting codes.
IEICE Trans. Fundamentals, E83-A(10), 1878--1883.



\bibitem{mullen2013}
Mullen, G. L., and Panario, D. (2013). Handbook of Finite Fields. CRC Press.

\bibitem{nebe2016}
Nebe, G., and Willems, W. (2016). On self-dual MRD codes. Adv. Math. Commun., 10(3), 633--642.

\bibitem{matsumoto2026}
Nizuka, R., and Matsumoto, R. (2026). Construction of quantum rank-metric codes using Hermitian orthogonality. arXiv:2605.02571.




\bibitem{ravagnani2016}
Ravagnani, A. (2016). Rank-metric codes and their duality theory. Des. Codes Cryptogr., 80(1), 197--216.

\bibitem{seroussi1980}
Seroussi, G., and Lempel, A. (1980). Factorization of symmetric matrices and trace-orthogonal bases in finite fields. SIAM J. Comput., 9(4), 758--767.

\bibitem{silva2008}
Silva, D., Kschischang, F. R., and Koetter, R. (2008). A rank-metric approach to error control in random network coding. IEEE Trans. Inform. Theory, 54(9), 3951--3967.

\bibitem{wan2003}
Wan, Z.-X. (2003). Lectures on Finite Fields and Galois Rings. World Scientific.

\end{thebibliography}
\end{document}